\documentclass[runningheads]{llncs}

\usepackage{latexsym}
\usepackage{amsmath}
\usepackage{amssymb}
\usepackage[table]{xcolor}
\usepackage{graphicx}
\usepackage{verbatim}
\usepackage{url}

\newcommand{\bbF}{\mathbb{F}}
\newcommand{\prob}[1]{\mathbb{P}\left[#1 \right]}
\newcommand{\cruc}{\mathcal{T}^{\geq 3}}
\newcommand{\mapnbit}{\colon \bbF_2^n \rightarrow \bbF_2}

\newcommand{\inceil}[1]{\left\lceil #1 \right\rceil}

\title{Constructive Relationships Between Algebraic Thickness and Normality
\thanks{Partially supported by the Danish
Council for Independent Research, Natural Sciences.
}}

\author{%
Joan Boyar\inst{1}  \and 
Magnus Gausdal Find\inst{2}\thanks{Most of this work was done while at the University of Southern Denmark}
}%
\institute{
Department of Mathematics of Computer Science, University of Southern Denmark\\
\email{joan@imada.sdu.dk}\\ 
\and
Information Technology Laboratory, National Institute of Standards and Technology,
USA\\
\email{magnus.find@nist.gov}
}

\begin{document}

\maketitle

\begin{abstract}
We study the relationship between two measures of Boolean functions;
\emph{algebraic thickness} and \emph{normality}. For a function $f$, the algebraic thickness
is a variant of the \emph{sparsity}, the number of nonzero coefficients in the unique
$\bbF_2$ polynomial representing $f$, and the normality is the largest dimension of
an affine subspace on which $f$ is constant.
We show that for $0 < \epsilon<2$, any function with algebraic thickness $n^{3-\epsilon}$ is constant on some affine subspace of dimension $\Omega\left(n^{\frac{\epsilon}{2}}\right)$.
Furthermore, we give an algorithm for finding such a subspace.
We show that this is at most a factor of $\Theta(\sqrt{n})$ from the best guaranteed, and when restricted to the technique used, is at most a factor of $\Theta(\sqrt{\log n})$ from the best guaranteed. We also show that a concrete function, majority, has algebraic thickness $\Omega\left(2^{n^{1/6}}\right)$.
 \end{abstract}

\section{Introduction and Known Results}
Boolean functions play an important role in many areas of computer science. In cryptology,
Boolean functions are sometimes classified according to some measure of complexity (also
called cryptographic complexity
\cite{DBLP:conf/dagstuhl/Carlet06}, nonlinearity criteria
\cite{DBLP:conf/eurocrypt/MeierS89} or nonlinearity measures
\cite{DBLP:conf/ciac/BoyarFP13}).
Examples of such measures are \emph{nonlinearity}, \emph{algebraic degree},
\emph{normality}, \emph{algebraic thickness} and \emph{multiplicative complexity},
and there are a number of results showing that  functions that are simple
according to a certain measure are vulnerable to a certain attack
(see \cite{carletbook} for a good survey).

A significant amount of work in this area presents explicit functions that achieve high (or low) values according to some measure. For the \emph{nonlinearity} measure
this was settled by showing the existence of bent functions
\cite{rothausbent}, for \emph{algebraic degree} the problem is trivial, for \emph{multiplicative complexity} this is a well studied problem in circuit complexity \cite{DBLP:journals/tcs/BoyarPP00}, for \emph{normality} this is exactly the problem of finding good \emph{affine dispersers} \cite{DBLP:conf/focs/Shaltiel11}. The first result in this paper is that the majority function has exponential algebraic thickness.

Another line of work has been to establish relationships between
these measures, e.g. considering questions of the form ``if a function $f$ is simple
(or complex) according to one measure, what does that
say about $f$ according to some other measure'', see e.g. \cite{carlet2001complexity,carletbook,DBLP:conf/ciac/BoyarFP13} and the references therein.
In this paper we focus on the relationship between \emph{algebraic thickness}
and \emph{normality}. Intuitively, these measures capture, each in their own way,
how ``far'' functions are from being linear
\cite{DBLP:journals/tit/Carlet04,DBLP:conf/dagstuhl/Carlet06}.
In fact, these two measures have been studied together previously (see e.g. \cite{carlet2003algebraic,DBLP:journals/tit/Carlet04}).
The relationship between these measures was considered in the work of Cohen and Tal
in \cite{DBLP:journals/corr/CohenT14}, where
they show that functions with a certain algebraic thickness have a certain
normality. 
For relatively small values of algebraic thickness, we tighten their bounds and present an algorithm to witness this normality.
The question of giving a constructive proof of normality is not just a theoretical one.
Recently a generic attack on stream ciphers with high normality was
successfully mounted in the work \cite{DBLP:journals/pmh/MihaljevicGPI12}.
If it is possible to constructively compute a witness of normality given a function with
low algebraic thickness, this implies that any function with low algebraic thickness
is likely to be vulnerable to the attack in \cite{DBLP:journals/pmh/MihaljevicGPI12}, as
well as any other attack based on normality. This work suggests that this is indeed possible
for functions with small algebraic thickness.

\section{Preliminaries and Known Results}
Let $\bbF_2$ be the field of order $2$, $\bbF_2^n$
the $n$-dimensional vector space over $\bbF_2$, and
$[n]=\{1,\ldots ,n \}$.
A mapping from $\bbF_2^n$ to $\bbF_2$ is called a \emph{Boolean function}. It is a well known fact that any Boolean function $f$ in the variables $x_1,\ldots ,x_n$ can be expressed uniquely as a multilinear polynomial over $\bbF_2$ called the \emph{algebraic normal form} (ANF) or the \emph{Zhegalkin polynomial}.
That is, there exist unique constants 
$c_{\emptyset},
\ldots , c_{\{1,\ldots ,n \}}$ over $\{ 0,1\}$, such that
\[
f(x_1,\ldots ,x_n) = \sum_{S\subseteq [n]}c_{S}\prod_{j\in S}x_j,
\]
where arithmetic is in $\bbF_2$. In the rest of this paper, most arithmetic will be
in $\bbF_2$, although we still need arithmetic in $\mathbb{R}$.
If nothing is mentioned it should be clear from the context what field is referred to.
The largest $|S|$ such that $c_S=1$ is
called the \emph{(algebraic) degree} of $f$, and functions with degree $2$ are called
\emph{quadratic} functions. We let $\log(\cdot)$ be the logarithm base two, $\ln(\cdot)$ the natural
logarithm, and $\exp(\cdot)$ the natural exponential function with base $e$.

\subsubsection{Algebraic Thickness}
For a Boolean function, $f$, let $\| f \| = \sum_{S\subseteq [n]}c_{S}$, with arithmetic in
$\mathbb{R}$. This measure is sometimes called the \emph{sparsity} of $f$ (e.g. 
\cite{DBLP:journals/corr/CohenT14}).
The \emph{algebraic thickness} \cite{carlet2001complexity,DBLP:journals/tit/Carlet04}
of $f$, denoted $\mathcal{T}(f)$ is defined as the smallest sparsity after any
affine bijection has been applied to the inputs of $f$. More precisely, letting
$\mathcal{A}_n$ denote the set of affine, bijective operators on $\bbF_2^n$, 
\begin{equation}
  \label{eq:algthickdef}
  \mathcal{T}(f) = \min_{A\in \mathcal{A}_n }
\|
f\circ A\|.
\end{equation}

Algebraic thickness was introduced and first studied by Carlet
in \cite{carlet2001complexity,carlet2003algebraic,DBLP:journals/tit/Carlet04}.
 Affine functions have algebraic thickness at most $1$, and Carlet showed that for any
constant $c>\sqrt{\ln 2}$, for sufficiently large $n$ there exist functions with algebraic
thickness
\[
 2^{n-1}-cn2^{\frac{n-1}{2}},
\]
and that a \emph{random} Boolean function will have such high algebraic thickness with high probability. Furthermore \emph{no} function has algebraic thickness larger
than $\frac{2}{3}2^{n}$.
Carlet observes that algebraic thickness was also implicitly mentioned in
\cite[Page 208]{DBLP:books/crc/MenezesOV96} and 
related to the so called ``higher order differential attack''
due to Knudsen \cite{knudsen1995truncated} and
Lai \cite{lai1994higher} in that they are dependent on the degree as well as the number
of terms in the ANF of the function used.

\subsubsection{Normality}
A $k$-dimensional \emph{flat} is an affine (sub)space of $\bbF_2^n$ with dimension $k$.
A function is $k$\emph{-normal} if there exists a $k$-dimensional flat
$E$ such that $f$ is constant on $E$ \cite{DBLP:journals/jc/Charpin04,carlet2001complexity}.
For simplicity define the \emph{normality} of a function $f$,
which we denote $\mathcal{N}(f)$, as the \emph{largest} $k$ such that $f$ is $k$-normal. 
We recall that affine functions have normality at least $n-1$ (which is the largest
possible for non-constant functions), while for any $c>1$, a random Boolean function
has normality less than $c\log n$
with high probability.

Functions with normality smaller than $k$ are often called 
\emph{affine dispersers} of dimension $k$, and a great deal of work has been put into
explicit constructions of functions with
low normality. Currently the asymptotically best known deterministic function, due to Shaltiel,
has normality less than $2^{\log^{0.9} n}$ \cite{DBLP:conf/focs/Shaltiel11}.

Notice the asymmetry in the definitions: linear functions have very low algebraic thickness
($0$ or $1$) but very high normality ($n$ or $n-1$), whereas random functions, with
high probability, have very
high algebraic thickness (at least $2^{n-1}-0.92\cdot n\cdot 2^{\frac{n-1}{2}}$)
but low normality (less than $1.01\log n$) \cite{carlet2003algebraic}.

\subsubsection{Remark on Computational Efficiency}
In this paper, we say that something is efficiently computable if it is computable
in time polynomially bounded in the size of the input.
Algorithms in this paper will have a Boolean function with a certain algebraic thickness
as input. We assume that the function is represented by the ANF of the
function witnessing this small algebraic thickness along with the bijection. That is, if a function $f$ with algebraic thickness $\mathcal{T}(f)=T$ is the input to the algorithm, we assume that it is represented by a function $g$ and an affine bijection $A$ such that
$g=f\circ A$ and $\| g \|=T$. In this setting, representing a
function $f$ uses $poly(\mathcal{T}(f)+n^2)$ bits.

\subsubsection{Quadratic Functions}
The normality and algebraic thickness of quadratic functions are well understood
due to the following theorem due to Dickson \cite{dicksonsthm} (see
also \cite{carletbook} for a proof).

\begin{theorem}[Dickson]
\label{thm:dickson}
 Let $f\colon \bbF_2^{n}\mapsto \bbF_2$ be quadratic.
Then there exist an invertible $n\times n$ matrix $A$, a vector
$\mathbf{b}\in \bbF_2^{n}$, $t\leq \frac{n}{2}$, 
and $c\in\bbF_2$ such that for $\mathbf{y}=A\mathbf{x}+\mathbf{b}$ one of the following two
equations holds:
\[
f(x)=y_1y_2+y_3y_4+\ldots y_{t-1}y_{t}+c, \textrm{ or } 
f(x)=y_1y_2+y_3y_4+\ldots y_{t-1}y_{t}+y_{t+1}.
\]
Furthermore $A$, $\mathbf{b}$ and $c$ can be found efficiently.
\end{theorem}

That is, any quadratic function is affine equivalent to some inner product function.
We highlight a simple but useful consequence of Theorem~\ref{thm:dickson}.
Simply by setting one variable in each of the degree two terms to zero, one gets:
\begin{proposition}
\label{prop:quadraticnormality}
Let $f\mapnbit $ be quadratic. Then $\mathcal{N}(f)\geq 
\left\lfloor \frac{n}{2}\right\rfloor$. Furthermore a
flat witnessing the normality of $f$ can be found efficiently.
\end{proposition}

\subsubsection{Some Relationships}
It was shown in \cite{DBLP:journals/tit/Carlet04} that
normality and algebraic thickness are logically independent of (that is, not subsumed by)
each other.
Several other results relating
algebraic thickness and normality to other cryptographic measures
 are given in \cite{DBLP:journals/tit/Carlet04}.
We mention a few relations to other measures.

Clearly, functions with degree $d$ have algebraic thickness $O( n^d )$, so having
superpolynomial algebraic thickness requires superconstant degree.
The fact that there exist functions with low degree and low normality has been
established in \cite{carlet2001complexity} and \cite{DBLP:journals/corr/CohenT14}
independently. In the following, by a \emph{random degree three polynomial}, we mean a function where each term of degree three is included in the ANF independently with probability $\frac{1}{2}$. No other terms are included in the ANF.

\begin{theorem}[\cite{carlet2001complexity,DBLP:journals/corr/CohenT14}\footnote{The constant $6.12$ does not appear explicitly in these articles, however it can be derived using similar calculations as in the cited papers. This also follows from Theorem~\ref{thm:randomsmallthick} later in this paper. We remark that $6.12$ is not optimal.}]
\label{thm:degree3disperser}
 Let $f$ on $n$ variables be a random degree three polynomial. Then with high probability,
$f$ remains nonconstant on any subspace of dimension $6.12\sqrt{n}$.
\end{theorem}

In fact, as mentioned in \cite{DBLP:journals/corr/CohenT14} it is not
hard to generalize this to the fact that for any constant $d$, a random degree $d$
polynomial has normality $O\left(n^{1/(d-1)}\right)$.
Perhaps surprisingly, this is tight. More precisely the authors give an elegant proof showing that \emph{any} function with degree $d$ has $\mathcal{N}(f)\in \Omega\left(n^{1/{(d-1)}}\right)$. 
This result implies the following relation between algebraic thickness
and normality.

\begin{theorem}[Cohen and Tal \cite{DBLP:journals/corr/CohenT14}]
\label{avishayGilAlgThickNorm}
Let $c$ be an integer and let $f$ have $\mathcal{T}(f)\leq n^c$.
Then $\mathcal{N}(f)\in \Omega\left(n^{1/(4c)}\right)$.
\end{theorem}

The proof of this has two steps: First they show by probabilistic methods that 
$f$ has a restriction with a certain number of free variables
and a certain degree, and after this they appeal to a relation between 
degree and normality. Although the authors do study the algorithmic question of finding such a subspace, they do not propose an efficient algorithm for finding a subspace of such dimension. 
We will pay special attention to the following type of restrictions of Boolean functions.

\begin{definition}
Let $f\mapnbit$. Setting $k<n$ of the bits to $0$ results in a new function $f'$ on
$n-k$ variables. We say that $f'$ is a $0$-restriction of $f$.
\end{definition}

By inspecting the proof in the next section and
the proof of Theorem~\ref{avishayGilAlgThickNorm}, one can see that \emph{most} of the restrictions performed are in fact setting variables to $0$.
Furthermore, by inspecting the flat used for the attack performed in
\cite{DBLP:journals/pmh/MihaljevicGPI12} (section 5.3), one can see that it is
of this form as well. Determining whether a given function represented by its ANF admits a $0$-restriction $f'$ on $n-k$ variables with $f'$ constant corresponds exactly to the hitting set problem, and this is well known  to be $\mathbf{NP}$ complete \cite{DBLP:books/fm/GareyJ79}. Furthermore it remains $\mathbf{NP}$ complete even when restricted to quadratic functions (corresponding to the vertex cover problem).

This stands in contrast to Proposition~\ref{prop:quadraticnormality}; for quadratic functions and general flats (as opposed to just $0$-restrictions) the problem is polynomial time solvable.
To the best of our knowledge, the computational complexity of the following problem is open (see also \cite{DBLP:journals/corr/CohenT14}): Given a function, represented by its ANF, find a large(st) flat on which the function is constant.

\section{Majority Has High Algebraic Thickness}
For many functions, it is trivial to see that the ANF contains many terms, e.g. the function
\[
 f(\mathbf{x})=(1+x_1)(1+x_2)\cdots (1+x_n),
\]
which is $1$ if and only if all the inputs are $0$, contains all the possible $2^n$ terms in its ANF. However, we are not aware of any explicit function along with a proof of a strong (e.g. exponential) lower bound on the algebraic thickness.
Using a result from circuit complexity \cite{razborov1987lower}, it is straightforward to show that the 
\emph{majority function}, $MAJ_n$ has exponential algebraic thickness. 
$MAJ_n$ is $1$ if and only if at least half of the $n$ inputs are $1$.
In the following, an $AC_0[\oplus]$ circuit of depth $d$ is a circuit with inputs
$x_1,x_2,\ldots ,x_n,(1\oplus x_1),(1\oplus x_2),\ldots ,(1\oplus x_n)$.
The circuit contains $\land,\lor,\oplus$ (AND, OR, XOR) gates of unbounded fan-in, and every directed path contains at most $d$ edges.
First we need the following simple proposition:

\begin{proposition}
 Let $f\mapnbit$ have $\mathcal{T}(f)\leq T$. Then $f$ can be computed by an $AC_0[\oplus]$ circuit of depth $3$ with at most $n+T+1$ gates.
\end{proposition}
\begin{proof}
 Suppose $f=g\circ A$ for some affine bijective mapping $A$.
In the first layer (the layer closest to the inputs) one can compute $A$ using $n$ XOR gates of fan-in at most $n$. Then by computing all the monomials independently, $g$ can be computed by an $AC_0[\oplus]$ circuit of depth $2$ using $T$ AND gates with fan-in at most $n$ and $1$ XOR gate of fan-in $T$. 
\qed
\end{proof}

Now we recall a result due to Razborov \cite{razborov1987lower}, see also \cite[12.24]{DBLP:books/daglib/0028687}

\begin{theorem}[Razborov]
 Every unbounded fan-in depth-d circuit over $\{\land,\lor,\oplus \}$ computing $MAJ_n$
requires $2^{\Omega(n^{1/(2d)})}$ gates.
\end{theorem}
Combining these two results, we immediately have the following result that the majority function $MAJ_n$ has high algebraic thickness.

\begin{proposition}
 $\mathcal{T}(MAJ_n)\geq 2^{\Omega(n^{1/6})}$.
\end{proposition}

\section{Algebraic Thickness and Normality}

This section is devoted to showing that functions with
algebraic thickness at most $n^{3-\epsilon}$ are constant on flats of somewhat large dimensions. Furthermore our proof reveals a polynomial time algorithm to find such a subspace.
In the following, a term of degree at least $3$ will be called a \emph{crucial} term, and for a function $f$, the number of crucial terms will be denoted $\cruc (f)$.

Our approach can be divided into two steps: First it uses $0$-restrictions
to obtain a quadratic function, and after this we can use  Proposition~\ref{prop:quadraticnormality}.
As implied by the relation between $0$-restrictions and the hitting set problem, finding the optimal $0$-restrictions is indeed a computationally hard task. Nevertheless, as we shall show in this section, the following greedy algorithm gives reasonable guarantees.

The greedy algorithm simply works by continually finding the variable that is contained in the most crucial terms, and sets this variable to $0$. It finishes when there are no crucial terms. We show that when the greedy algorithm finishes, the number of variables left, $n'$, is relatively large as a function of $n$ (for a more precise statement, see Theorem~\ref{thm:subbatov}). Notice that we are only interested in the behavior of $n'$ as a function of $n$, and that this is not necessarily related to the approximation ratio of the greedy algorithm, which is known to be $\Theta(\log n)$ \cite{DBLP:journals/jcss/Johnson74a}.

We begin with a simple proposition about the greedy algorithm that will be useful throughout the section, and it gives a tight bound.
\begin{proposition}
\label{prop:greedyLB}
  Let $g\mapnbit$ have $\cruc(g)\geq m$. Then some variable
$x_j$ is contained in at least $\left\lceil 3\frac{m}{n}\right\rceil$
crucial terms.
\end{proposition}
\begin{proof}
We can assume that no variable occurs twice in the same term. Hence the total number of variable occurrences in crucial terms is at least $3m$. By the pigeon hole principle, some variable is contained in at least $\left\lceil 3\frac{m}{n}\right\rceil$ terms.
\qed
\end{proof}

 The following lemma is a special case where a tight result
can be obtained. It is included here because the result is
tight, and it gives a better constant in Theorem~\ref{thm:subbatov}
than one would get by simply removing terms one at a time. 
The result applies to functions with relatively small
thickness, and a later lemma reduces functions with somewhat larger
thickness to this case.
 \begin{lemma}
 \label{bootstraplemma}
 Let $c\leq \frac{2}{3}$ and let $f\mapnbit$ have
 $\cruc (f)\leq cn$.
 Then $f$ has a $0$-restriction $f'$
 on $n'= n-\inceil{ \frac{3c-1}{5}n }$ variables
 with $\cruc (f')\leq \frac{n'}{3}$.
 \end{lemma}

 \begin{proof}
 Let the greedy algorithm run until a function $f'$ on $n'$ variables with
 $\cruc(f')\leq \frac{n'}{3}$ is obtained. 
 By Proposition~\ref{prop:greedyLB} we eliminate at least $2$ terms
 in each step.
 The number of algorithm iterations is at most $\inceil{\frac{3c-1}{5}n}$.
 Indeed, let
 \(
 \inceil{ \frac{3c-1}{5}n }=\frac{3c-1}{5}n+\delta
 \)
 for some
 \(
 0\leq \delta< 1
 \).
 After this number of iterations the number of variables left is
 \[
         n'=n-\frac{3c-1}{5}n-\delta=\frac{6-3c}{5}n-\delta
 \]
 and the number of critical terms is at most 
 \[
 cn-2\left(\frac{3c-1}{5}n-\delta\right)
 =
  \frac{2-c}{5}n-2\delta.
 \]
 In particular $\frac{n'}{3}\geq \frac{2-c}{5}n-2\delta$.
\qed
 \end{proof}

 Lemma~\ref{bootstraplemma} is essentially tight.

\begin{proposition}
 Let $\frac{1}{3}< c\leq \frac{2}{3}$ be arbitrary but rational. Then for infinitely many values of $n$, there exists a function on $n$ variables with $\cruc(f)=cn$ such that every $0$-restriction $f'$ on $n'>n-\left\lceil \frac{3c-1}{5}n\right\rceil $ variables has $\cruc(f)> \frac{n'}{3}$.
\end{proposition}
\begin{proof}
 Let $\frac{1}{3}< c\leq \frac{2}{3}$ be fixed and consider the function on $6$ variables:
 \[
 f(x)=x_1 x_2 x_3 + x_1 x_4 x_5 + x_2 x_4 x_6 + x_3 x_5 x_6.
 \]
The greedy algorithm sets this functions to $0$ by killing two variables, and this is optimal. Furthermore setting any one variable to $0$ kills exactly two terms.
Now consider the following function defined on $n=30m$ variables and
having $20m$ terms. For
 convenience we index the variables by $x_{i,j}$ for $1\leq i\leq 5m$, $1\leq j\leq 6$. Let
 \[
  g(x)=\sum_{i=1}^{5m}f(x_{i,1},x_{i,2},x_{i,3},x_{i,4},x_{i,5},x_{i,6}).
 \]
Again here the greedy algorithm is optimal, and setting $6m$ variables
to zero leaves $n'=24m$ variables and $8m$ terms remaining.
Thus, the bound from Lemma~\ref{bootstraplemma} is met with equality for $c=\frac{2}{3}$.

To see that it is tight for $c<\frac{2}{3}$, consider the function, $\tilde{f}$ on $n$ variables, where $n$ is a multiple of $30$ such that 
$c\frac{4}{3}\frac{n}{2-c}$ is an integer. Run the greedy algorithm until the number of variables is $\tilde{n}$ and $\cruc(\tilde{f})=c\tilde{n}$ (assuming $c\tilde{n}$ is an integer). At this point $\tilde{n}=\frac{4}{3}\frac{n}{2-c}$ and the number of terms left is $c\tilde{n}$. Again, by the structure of the function, setting any number, $t$, of
the variables to $0$ results in 
a function with $\tilde{n}-t$ variables and at least $c\tilde{n}-2t$ terms. When $t<\frac{(3c-1)\tilde{n}}{5}$, we have $c\tilde{n}-2t>\frac{\tilde{n}-t}{3}$.

\qed
 
\end{proof}

An immediate corollary to Lemma~\ref{bootstraplemma} is the following.

\begin{corollary}
\label{cor:combiner}
 Let $f\mapnbit$ have $\cruc (f)\leq \frac{2}{3}n$. Then it is constant on a flat of dimension $n'\geq \left\lfloor \frac{\frac{2}{3}\left\lfloor \frac{4}{5}n\right\rfloor}{2}\right\rfloor\geq \frac{4}{15}n-2$. Furthermore, such a flat can be found efficiently.
\end{corollary}
\begin{proof}
 First apply Lemma~\ref{bootstraplemma} to obtain a function on $n'= \left\lfloor \frac{4}{5}n\right\rfloor$ variables with at most $\frac{n'}{3}$ crucial terms. Now set one variable in each crucial term to $0$, so after this we have at least $\frac{2}{3}\left\lfloor \frac{4}{5}n\right\rfloor$ variables left and the remaining function is quadratic. Applying 
Theorem~\ref{thm:dickson} gives the result.\qed
\end{proof}

The following lemma generalizes the lemma above to the case with more terms. The analysis of the greedy algorithm uses ideas similar to those used in certain formula lower bound proofs, see e.g. \cite{subbotovskaya1961realizations} or \cite[Section 6.3]{DBLP:books/daglib/0028687}.
\begin{lemma}
\label{lem:subbotovstyle}
 Let $f\mapnbit$ with $\cruc(f)\leq n^{3-\epsilon}$, for $0<\epsilon<2$. Then there exists a $0$-restriction
$f'$ on $n'=\left\lfloor \sqrt{\frac{2}{3}n^{\epsilon}}\right\rfloor$ variables with $\cruc(f')\leq \frac{2}{3}n'$.
\end{lemma}
\begin{proof}
 Let $\cruc(f)=T$. Then, by Proposition~\ref{prop:greedyLB}. Setting the variable contained  in the largest number of terms to $0$, the number of crucial terms left is at most
\[
 T-\frac{3T}{n}=T\cdot \left(1-\frac{3}{n}\right)\leq T\cdot\left(\frac{n-1}{n}\right)^3.
\]
Applying this inequality $n-n'$ times yields that after $n-n'$ iterations the number of crucial terms left is at most
\[
 T\cdot\left(\frac{n-1}{n} \right)^3\left(\frac{n-2}{n-1} \right)^3
\cdots
\left(\frac{n'}{n'+1} \right)^3
=
T\cdot \left( \frac{n'}{n}\right)^3.
\]
When $n'=\sqrt{\frac{2}{3}n^{\epsilon}}$ and $T=n^{3-\epsilon}$, this is at most
$\frac{2}{3}n'$.
\qed
\end{proof}
\emph{Remark:} A previous version of this paper \cite{DBLP:journals/corr/BoyarF14a}, contained a version of the lemma with a proof substantially more complicated. We thank anonymous reviewers for suggesting this simpler proof.

It should be noted that Lemma~\ref{lem:subbotovstyle} cannot be improved to the case where $\epsilon=0$, no matter what algorithm is used to choose the $0$-restriction. To see this consider the function containing all degree three terms. For this function, \emph{any} $0$-restriction (or $1$-restriction) leaving $n'$ variables will have at least $\binom{n'}{3}$ crucial terms. 
On the other hand,  restricting with $x_1+x_2=0$ results in all crucial terms
with 
both $x_1$ and $x_2$  having lower degree and all crucial terms with just one
of them cancelling out.
This suggests that for handling functions with larger algebraic thickness, one should use restrictions other than just $0$-restrictions.

Combining Lemma~\ref{lem:subbotovstyle} with Corollary~\ref{cor:combiner}, we get the following theorem. 

\begin{theorem}
\label{thm:subbatov}
 Let $\mathcal{T}(f)=n^{3-\epsilon}$ for $0<\epsilon<2$. Then there exists a flat of dimension
at least
$\frac{4}{15}\sqrt{\frac{2}{3}n^{\epsilon}}-3$, such that when restricted to this flat, $f$ is constant. Furthermore this flat can be found efficiently.
\end{theorem}

This improves on Theorem~\ref{avishayGilAlgThickNorm} for functions
with algebraic thickness $n^{s}$ for $1\leq s\leq 2.82$, and the smaller $s$, the bigger
the improvement, e.g. for $\mathcal{T}(f)\leq n^{2}$, our bound guarantees
$\mathcal{N}(f)\in \Omega(n^{1/2})$, compared to $\Omega(n^{1/8})$.

\subsection{Normal Functions with low sparsity}
How good are the guarantees given in the previous section?
The purpose of this section is first to show that the result from
Theorem~\ref{thm:subbatov} is at most a factor of $\Theta(\sqrt{n})$
from being tight. More precisely, we show that for any $2<s\leq 3$ there exist
functions with thickness at most $n^s$ that are nonconstant on flats
of dimension $O(n^{2-\frac{s}{2}})$.
Notice that this contains Theorem~\ref{thm:degree3disperser} as a special case where $s=3$.

\begin{theorem}
\label{thm:randomsmallthick}
For any $2<s\leq 3$, for sufficiently large $n$,
there exist functions with degree $3$ and algebraic thickness at most
$n^{s}$ that, for sufficiently large $n$, remain nonconstant on all flats of dimension
$6.12n^{2-\frac{s}{2}}$.
\end{theorem}
\begin{proof}
The proof uses the probabilistic method. We endow the set of all
Boolean functions of degree $3$ with a probability distribution
$\mathcal{D}$, 
and show that under this distribution a function has the
promised normality with high probability.

The proof is divided into the following steps:
First we describe the probability distribution
$\mathcal{D}$. Then, we fix an arbitrary $k$-dimensional flat
$E$, and bound the probability that a random $f$ chosen according to
$\mathcal{D}$ is constant on $E$.
We show that for $k=Cn^{2-s/2}$, where the constant $C$ is determined later,
this probability is sufficiently small that a
union bound over all possible choices of $E$ gives the desired result.

We define $\mathcal{D}$ by describing the probability distribution on the ANF. We let each possible degree $3$ term
be included with probability $\frac{1}{2n^{3-s}}$. The expected number
of terms is thus $\frac{1}{2}n^{s-3}\binom{n}{3}\leq n^{s}/12$, and the probability
of having more than $n^{s}$ terms is less than $0.001$ for large $n$. 
Now let $E$ be an arbitrary but fixed $k$-dimensional flat.

One  way to think of a function restricted to a $k$-dimensional
flat is that it can be obtained by a sequence of $n-k$ affine variable substitutions
of the form $x_i := \sum_{j\in S}x_j+c$. This changes the ANF of
the function since $x_i$ is no longer a ``free'' variable.
Assume without loss of generality that we substitute for the variables
$x_n,\ldots ,x_{k+1}$ in that order. Initially we start with
the function $f$ given by
\[
f(x) = \sum_{\{ a,b,c\}\subseteq [n]} I_{abc}x_ax_bx_c,
\]
where $I_{abc}$ is the indicator random variable, indicating
whether the $x_ax_bx_c$ is contained in the ANF.
Suppose we perform the $n-k$ restrictions and obtain the function $\tilde{f}$.
The ANF of $\tilde{f}$ is given by
\[
 f(x)
=
\sum_{\{a , b , c \} \subseteq [k]}
\left(I_{abc}+ \sum_{s\in S_{abc}}I_{s}\right)x_ax_bx_c,
\]
where $S_{abc}$ is some set of indicator random variables depending on the
restrictions performed.
It is important that $I_{abc}$, the indicator random variable corresponding to
$x_ax_bx_c$, for $\{a,b,c \}\subseteq [k]$ is \emph{only} occurring at $x_ax_bx_c$. Hence
we conclude that independently of
the outcome of all the indicator random variables $I_{a'b'c'}$ with
$
\{a',b',c' \}\not\subseteq [k],
$
we have that the marginal probability for any $I_{abc}$ with $\{a,b,c \}\subseteq [k]$
occurring remains at least $\frac{1}{2n^{3-s}}$.

Define $t=\binom{k}{3}$ random variables,
 $Z_1,\ldots ,Z_t$,
 one for each potential term in the ANF of $\tilde{f}$, such that $Z_j=1$ if and
 only if the corresponding term is present in the ANF, 
 and $0$ otherwise.
The obtained function is only constant if there are no degree 3 terms,
so the probability of $\tilde{f}$ being constant is thus at most
\begin{align*}
\prob{
Z_1=\ldots= Z_t =0
}
\leq &
\left(1-\frac{1}{2n^{3-s}} \right)^{\binom{k}{3}}\\
\leq&
\left(1-\frac{1}{2n^{3-s}} \right)^{\frac{C^3}{27} ( n^{6-3s/2}  )  }\\
=&
\left(\left(1-\frac{1}{2n^{3-s}} \right)^{2n^{3-s}}\right)^{\frac{C^3}{54}(n^{3-s/2} )}\\
\leq &
\exp\left(-\frac{C^3}{54}n^{3-s/2}\right).
\end{align*}
The number of choices for $E$ is at most $2^{n(k+1)}$, so the probability
that $f$ becomes constant on \emph{some} affine flat of dimension $k$
is at most
\begin{align*}
\exp\left(-\frac{C^3}{54}n^{3-s/2} + C\ln(2)    n^{3-s/2}   +n)\right).
\end{align*}
Now if $C>\sqrt{54\ln(2)}\approx 6.11..$, this quantity tends to $0$. We 
conclude that with high probability the function obtained has algebraic thickness
at most $n^s$ and normality at most $6.12n^{2-\frac{s}{2}}$.
\qed
\end{proof}

There is factor of $\Theta(\sqrt{n})$ between the existence guaranteed by Theorem~\ref{thm:subbatov} and Theorem~\ref{thm:randomsmallthick} and we leave it as an interesting problem to close this gap.

The algorithm studied in this paper works by setting variables to $0$
until all remaining terms have degree at most $2$, and after that appealing to 
Theorem~\ref{thm:dickson}.
A proof similar to the previous 
shows that among such algorithms, the bound from
Theorem~\ref{thm:subbatov} is very close to being asymptotically tight.

\begin{theorem}
\label{thm:randpolyandzerorestrictions}
For any $2<s<3$, there exist functions with degree $3$ and algebraic thickness at most
$n^{s}$ that have degree $3$ on any $0$-restriction of dimension $3\sqrt{\ln n}n^{\frac{3-s}{2}}$.  
\end{theorem}
\begin{proof}
  We use the same proof strategy as in the proof of
Theorem~\ref{thm:randomsmallthick}. Endow the set of all Boolean functions of degree 3
with the same probability distribution $\mathcal{D}$. For large $n$, the number of
terms is larger than $n^s$ with probability at most $0.001$.
Now we set all but
$C\sqrt{\ln n}n^{\frac{3-s}{2}}$ of
the variables to $0$, and consider the probability of the
function being constant under this fixed $0$-restriction. We will show that this probability
is so small that a union bound over all such choices gives that with high probability
the function is nonconstant under \emph{any} such restriction.
We will see that setting $C=3$ will suffice.
There are $\binom{C\sqrt{\ln n}n^{\frac{3-s}{2}}}{3}$ possible degree 3 terms on these
remaining variables, and we let each one be included with probability $\frac{1}{n^{3-s}}$.
The probability that none of these degree three terms are included is

\begin{align*}
\left(
1-\frac{1}{n^{3-s}}
\right)
^{
  \binom{C\sqrt{\ln n}n^{\frac{3-s}{2}}}{3}
}
&\leq
\left(
  1-\frac{1}{n^{3-s}}
\right)^{\frac{C^3}{27} \left(\sqrt{ \ln n}\right)^3 n^{\frac{9-3s}{2}}  }\\
&=
\left(
\left(
  1-\frac{1}{n^{3-s}}
\right)^{
n^{3-s}}
\right)^{
n^{\frac{3-s}{2}}\left(\ln n \right)^{3/2}\frac{C^3}{27}}
\\
 &\leq
 \exp \left(- \frac{C^3}{27} n^{\frac{3-s}{2}}(\ln n)^{3/2} \right),
\end{align*}
and the number of $0$-restrictions with all but $C\sqrt{\ln n}n^{\frac{3-s}{2}}$
variables fixed is
\begin{align*}
  \binom{n}{C\sqrt{\ln n}n^{\frac{3-s}{2}}}
  &
  \leq
	\frac{
		n^{C\sqrt{\ln n}n^{\frac{3-s}{2}}}
		}{
		(C\sqrt{\ln n}n^{\frac{3-s}{2}})!
		}
\\
&=
\exp (\ln n C\sqrt{\ln n}n^{\frac{3-s}{2}}
-
 \ln ((C\sqrt{\ln n}n^{\frac{3-s}{2}}) !))\\
&\leq
\exp \left(\ln^{3/2}( n )C n^{\frac{3-s}{2}}
-
.99C\sqrt{\ln n} n^{\frac{3-s}{2}}
\ln \left(C\sqrt{\ln n} n^{\frac{3-s}{2}}\right)
\right)
\\
&\leq
\exp \left(   (\ln n)^{3/2}   Cn^{\frac{3-s}{2}}
-
\frac{3-s}{2}.98C(\ln n)^{3/2}n^{\frac{3-s}{2}}
\right)\\
&=
\exp \left((\ln n)^{3/2}Cn^{\frac{3-s}{2}} \left(1 - 0.98\frac{3-s}{2}\right) 
\right),
\end{align*}
where the last two inequalities hold for sufficiently large $n$. Again, by the union bound, the probability that there exists such a choice on which there are no terms of degree three left is at most
\[
\exp \left(- \frac{C^3}{27} n^{\frac{3-s}{2}}(\ln n)^{3/2} \right)
 \exp \left((\ln n)^{3/2}Cn^{\frac{3-s}{2}} \left(1 - 0.98\frac{3-s}{2}\right) 
\right).
\]

For $C\geq 3$ this probability tends to zero, hence we have that with high probability the function does not have a $0$-restriction on $3\sqrt{\ln n}n^{\frac{3-s}{2}}$ variables of degree smaller than $3$.\qed
\end{proof}

%
%

\bibliographystyle{splncs03}

\bibliography{refs}

\end{document}